\documentclass{article}
\usepackage{kr}

\usepackage{times}
\usepackage{soul}
\usepackage{url}
\usepackage[hidelinks]{hyperref}
\usepackage[utf8]{inputenc}
\usepackage[small]{caption}
\usepackage{graphicx}
\usepackage{amsmath}
\usepackage{amsthm}
\usepackage{booktabs}
\usepackage{algorithm}
\usepackage{thm-restate}

\usepackage[dvipsnames]{xcolor}
\usepackage{bm,xspace}
\usepackage{amsfonts}

\usepackage{ amssymb }
\usepackage{ stmaryrd }
\usepackage{mathtools}
\usepackage{tikz}
\usetikzlibrary{shapes.misc, fit, decorations.pathreplacing,calligraphy,positioning,calc} 
\usetikzlibrary{automata}
\usetikzlibrary{decorations.pathmorphing,arrows.meta}
\usepackage[inline]{enumitem}

\usepackage{macros}

\newtheorem{theorem}{Theorem}

\makeatletter
\newsavebox{\@brx}
\newcommand{\llangle}[1][]{\savebox{\@brx}{\(\m@th{#1\langle}\)}%
  \mathopen{\copy\@brx\kern-0.5\wd\@brx\usebox{\@brx}}}
\newcommand{\rrangle}[1][]{\savebox{\@brx}{\(\m@th{#1\rangle}\)}%
  \mathclose{\copy\@brx\kern-0.5\wd\@brx\usebox{\@brx}}}
\makeatother

\usepackage{algorithm} 
\usepackage{algpseudocode}
\algnewcommand\algorithmiccase{\textbf{case}}
\algdef{SE}[CASE]{Case}{EndCase}[1]{\algorithmiccase\ #1}{\algorithmicend\ \algorithmiccase}%
\algtext*{EndCase}
\algtext*{EndIf}
\algtext*{EndFor}
\algtext*{EndProcedure}

\makeatletter
\newenvironment{breakablealgorithm}
  {
   \begin{center}
     \refstepcounter{algorithm}
     \hrule height.8pt depth0pt \kern2pt
     \renewcommand{\caption}[2][\relax]{
       {\raggedright\textbf{\fname@algorithm~\thealgorithm} ##2\par}%
       \ifx\relax##1\relax 
         \addcontentsline{loa}{algorithm}{\protect\numberline{\thealgorithm}##2}%
       \else 
         \addcontentsline{loa}{algorithm}{\protect\numberline{\thealgorithm}##1}%
       \fi
       \kern2pt\hrule\kern2pt
     }
  }{
     \kern2pt\hrule\relax
   \end{center}
  }
\makeatother

\title{I Would If I Could: 
\\ Reasoning about Dynamics of  Actions in Multi-Agent Systems\thanks{This is an extended version of the paper with the same title that will appear in KR 2026, and which contains a technical appendix with proof details.}}

\author{%
Rustam Galimullin$^1$\and
Hermine Grosinger$^2$\and
Munyque Mittelmann$^{3,4}$\\
\affiliations
$^1$University of Bergen, Bergen, Norway\\
$^2$Örebro University, Örebro, Sweden\\
$^3$CNRS, UMR 7030, F-93430, Villetaneuse, France \\
$^4$Université Sorbonne Paris Nord, LIPN, F-93430, Villetaneuse, France
\emails
rustam.galimullin@uib.no,
hermine.grosinger@oru.se,
mittelmann@lipn.univ-paris13.fr
}

\begin{document}

\maketitle

\begin{abstract}
Autonomous agents acting in realistic Multi-Agent Systems (MAS) should be able to adapt during their execution.  Standard strategic logics, such as Alternating-time Temporal Logic (\ATL), model agents' state- or history-dependent behaviour. However, the dynamic treatment of agents' available actions and their knowledge of required actions is still rarely addressed. In this paper, we introduce \ATL with Dynamic Actions (\ATLD), which models the process of granting and revoking actions, and its extension \ATELD, which captures how such updates affect agents’ knowledge.  Beyond the conceptual contribution, we provide several technical results: we analyse the expressivity of our logic in relation to \ATL, study its relation to normative systems, and provide complexity results for relevant computational problems. 
\end{abstract}


\section{Introduction}

In real-world, long-term scenarios with autonomous agents such as robots, it is unrealistic to assume that they possess all the necessary abilities and knowledge required to achieve certain tasks~\cite{grosinger2025next}. In general, what the agents should or should not be able to do might change dynamically. Reasoning about the skills an agent must acquire to achieve a given goal is naturally relevant for their design or reconfiguration. 
At the same time, formal treatment of dynamic phenomena in Multi‑Agent Systems (MAS) can be combined with verification techniques to assess the correctness and safety of such systems \cite{DynamicTurn}. %
Take the following example, adapted from~\cite{grosinger2025next}: 
%
%

\begin{quote}Bob is an elderly person. To support his everyday life at home, he has two autonomous robots. They help him, for example, to take his medicine as he is forgetful.  Robot 1 is tasked with retrieving the medicine from a drawer, while Robot 2 takes care of other tasks and thus does not have the action for opening the drawer. 
Now, suppose that the drawer with the medicine is stuck, and hence Robot 1 cannot get the medicine on her own. 
We want to \emph{grant} Robot 2 the same action, then the two robots together can open the drawer and thus deliver the medicine to Bob. Or if Bob is prescribed a new medicine that is sensitive to heat, we want to \emph{remove} the action from the robots to bring it into a warm area. Ultimately, we would like to be able to deal with situations when the robots are uncertain about which medicine they are handling and their precise location, 
while being able to reason about the evolution of their skills. 
\end{quote}



Research on MAS 
has a long tradition of 
investigating how to formally model problems arising from the (possibly adversarial) interactions of 
multiple autonomous agents operating within the same environment. 
%
%
%
In logic-based formalisms for strategic reasoning in MAS, such as  \textit{Alternating-time Temporal Logic} ($\ATL$) ~\cite{alur2002}, it is possible to express that an agent (or a coalition of agents) can bring about a temporal goal. \textit{Alternating-time Temporal Epistemic Logic} ~\cite{HoekW03a,JamrogaH04}, denoted $\ATEL$, 
extends $\ATL$ with epistemic operators, capturing agents' knowledge about the system and their ability to cooperate in imperfect information scenarios.  $\ATL$ and $\ATEL$ have been effectively implemented in tools such as MCMAS \cite{LomuscioQR17}, STV \cite{KurpiewskiJK19}, and VITAMIN \cite{0001M25}, and both logics have been used for verification of a wide variety or real systems, including e-voting protocols \cite{JamrogaMMPPSK22}, autonomous submarines \cite{EzekielLMV11}, and smart contracts \cite{NamK22}, to name a few.  

As expressive as $\ATL$ and $\ATEL$ may be, they cannot naturally model Bob's example  above. While $\ATL$ agents can dynamically react based on the state or history of the system execution, these logics are defined over fixed models. 
In this paper, we are interested in moving 
from a static representation of models to a more dynamic one, in which we are able to update the set of actions that an agent has at her disposal (e.g., an action to open a drawer) and ensure the agent knows about this action update. 
%
More precisely, we want to capture that agents \emph{know} which actions they have themselves and which actions the other ones have. Moreover, we want them to \emph{know} which actions they lack and which they need to acquire or drop in order to achieve a goal. Going back to  Bob's example, although the robots know that they currently cannot deliver the medicine, they also know that if Robot $2$ was granted --- or \textit{acquired} --- 
the action of opening the drawer, they can achieve the goal together. This type of reasoning is central in 
what \cite{grosinger2025next} calls \emph{proactive learning}.

\subsubsection{Contribution}
To address this limitation, we first introduce the new logic $\ATLD$
 which extends $\ATL$ (Section~\ref{sec:ATLD}), with update operators to grant and revoke actions. 
For the new logic and its fragments, we study their expressivity, and their relationship to some approaches in the field of normative reasoning in MAS \cite{AlechinaLD18}. Indeed, revoking actions, such as moving to a warm place in Bob's example, serves as a mechanism for regulating what agents are permitted to do.  Furthermore, we show that, as for $\ATL$, model-checking $\ATLD$ is in $\PTIME$ and that the problem of determining the existence of a bounded update such that a given $\ATLD$ formula is satisfied is \textsc{NP}-complete.  
 

 Next, we explore the epistemic dimension and introduce $\ATELD$ (Section~\ref{sec:ATELD}), 
which allows us not only to update agents' actions, but also capture the epistemic consequences of such updates. It enables reasoning about the agents that know which actions they have  and which actions they are lacking. 
Model-checking of $\ATELD$ is $\Delta^{\mathrm{P}}_2$-complete, which is the same as for $\ATEL$. Finally, we briefly discuss some related work and conlude in Section \ref{discussion}.


\section{ATL with Dynamic Actions}\label{sec:ATLD} 

We start by considering the setting with perfect information. 

\subsection{Syntax and Semantics}

We fix a countable set of atomic \textit{propositions} $\APf = \{p,q,...\}$, and finite sets  of \textit{agents} $\Ag = \{\ag, \agb, \dots\}$ and \textit{actions} $\Act = \{\alpha, \beta, ...\}$.  
The syntax of \textit{Alternating-time Temporal Logic with Dynamic actions} (\ATLD) extends that of \ATL \cite{alur2002} with operators for granting and revoking actions, as well as a construct to express that an agent currently has a particular action.

\begin{definition}[\ATLD]\label{def:ATLD-syntax}
 Formulas $\varphi$ of  \ATLD
are defined by the following grammar:
   \begin{align*}
    	\varphi  ::= &  p   \mid \neg \varphi \mid  (\varphi \lor \varphi) \mid  
        \coop{\coalition} \X \varphi \mid  
        \coop{\coalition} \varphi \until \varphi \mid  
        \coop{\coalition} \varphi \release \varphi \mid  
        \\ 
        &  [\pi]^+\varphi \mid  [\pi]^-\varphi  
         \mid
        \hasAction{\ag}{\act}     
        \\
      \pi ::= &   \; \varphi: \act \to \coalition \mid  \pi , \pi 
   \end{align*}
	where $p\in\APf$, $\act \in \Act$, $\ag \in \Ag$,  and $\coalition \subseteq \Ag$. 
    The fragments of \ATLD without $[\pi]^+\varphi$ or $[\pi]^-\varphi$ will be denoted as $\ATLD^-$ and $\ATLD^+$, respectively.  
    Finally, the fragment of \ATLD without $[\pi]^+\varphi$, $[\pi]^-\varphi$, and $\hasAction{\ag}{\act}$ is \ATL\footnote{We define \ATL using $\coop{\coalition} \varphi \release \varphi$ as it is strictly more expressive than \ATL defined using $\coop{\coalition} \textbf{G} \varphi$ \cite{Laroussinie2008}.}. 
 
\end{definition}

Operator $\coop{\coalition} \X \varphi$ 
means that a coalition $\coalition$ can force $\varphi$ to be true in the \textit{ne}\textbf{X}\textit{t} state. 
The operator $\coop{\coalition} \varphi \until \psi$ means that $\coalition$ can maintain $\varphi$ \textbf{U}\textit{ntil} $\psi$ becomes true. 
$\coop{\coalition} \varphi \release \psi$ means that $\coalition$ can maintain $\psi$ until $\varphi$ \textbf{R}\textit{eleases} the requirement for the truth of $\psi$. We can also define the usual derived temporal operators for \textit{eventually} and \textit{always} as follows: 
${\coop{\coalition} \F \varphi \colonequals \coop{\coalition}\top \U \varphi}$ and ${\coop{\coalition}\G \varphi \colonequals \coop{\coalition}\bot \release \varphi}$.  

Constructs $[\pi]^+\varphi$ mean that after applying the \emph{granting actions update} $\pi$, $\varphi$ holds. Formulas $[\pi]^-\varphi$ means that after applying the \emph{removing actions update} $\pi$, $\varphi$ holds.  A granting actions update (resp. a removing actions update) $\varphi: \act \to \coalition$ specifies that the action $\act$ is granted to (resp. removed from) the agents in $\coalition$ in \textit{every state} where $\varphi$ is true. Updates of different kinds can be applied sequentially (i.e., by chaining update operators), whereas updates of the same type can be applied simultaneously (e.g., $[\pi_1,\pi_2]^+\varphi$).  
Given an update $[\varphi_1: \act_1 \to \coalition_1, \dots, \varphi_k: \act_k \to \coalition_k ]^+$ (resp. $[\varphi_1: \act_1 \to \coalition_1, \dots, \varphi_k: \act_k \to \coalition_k ]^-$) we will abbreviate it as $[\overline{\varphi: \act \to \coalition}]^+_k$ (resp. $[\overline{\varphi: \act \to \coalition}]^-_k$). 
Operator $has(\ag,\act)$ means that in the current state agent $\ag$ has action $\act$ available to her. 
Finally, all standard abbreviations of propositional logic, like $\varphi \land \psi$ and $\varphi \to \psi$, and conventions for removing parentheses hold. 
Given a formula $\varphi$, we define its \emph{size} $|\varphi|$ as the number of symbols occurring in $\varphi$.

Now we turn our focus to the semantics of \ATLD.  

\begin{definition}[\CGS]
  \label{def-cgs}
A \emph{ concurrent game structure } (\CGS) is a tuple
$\CGS=(\setpos, \val, \availability, \transCGS)$, where $\setpos$ is a  finite set of \emph{states}; $\val:\setpos\to 2^\APf$ is a \emph{labelling function}; $\availability:\Ag \times \setpos \to 2^\Act \setminus \emptyset$ is an \emph{availability function}, defining a non-empty set of actions available to agents at each state; $\transCGS: \setpos \times \Act^{\Ag} \to V$ is a \emph{transition function}, assigning the outcome state to each state and all possible joint actions. We will also sometimes call a CGS a \emph{model}.
  
  Given a CGS $\CGS$ and a state $\pos$, let $True(v) = \{p \in \APf \mid p \in \val(v)\}$ be the set of propositions true in $\pos$.  
  We define the \emph{size} of $\CGS$, denoted $|\CGS|$, as $|\setpos| + \sum_{\pos \in \setpos}|True(v)|+|d| + |o|$. We call $\CGS$ \emph{finite}, if $|\CGS|$ is finite. 
\end{definition}

Observe that our definition of a \CGS slightly deviates from the standard one as we specify all possible transitions in $\transCGS$ even though agents may not have available actions in $\availability$ to execute all of them. We can consider such transitions as `enabled' and `disabled', and below we make it clear how the dynamic operators of granting and removing actions affect the set of transitions enabled in a \CGS.

\begin{definition}[Paths and Strategies]
    A \emph{path} $\iplay=\iplay_0 \iplay_1 \ldots \in \setpos^{\omega}$ 
is an infinite sequence of states  s.t. for every $i \geqslant 0$ there exists an action profile $\jmov \in \prod_{\ag \in \Ag}\availability(\ag,\iplay_{i})$ such that $\trans(\iplay_{i}, \jmov)=\iplay_{i + 1}$. 

A \emph{memoryless  strategy} for an agent $\ag$
is a function $\sigma_\ag: \setpos \to \Act$. 
We let $\setstrata$ be the set of memoryless  strategies for agent $\ag$, and 
$\setstrat_{\coalition}=\prod_{\ag \in \coalition}\setstrat_{\ag}$ be the set of collections of strategies for agents $\ag$ in $\coalition$. 

For a state $\pos$ and a  coalition $\coalition \in 2^\Ag \setminus \emptyset$, a \emph{strategy profile of $\coalition$} is $\profile\strat_\coalition \in \setstrat_{\coalition}$. 
The \emph{outcome function} $\out(\pos, \profile\strat_\coalition)$ returns the set of all paths starting at state $\pos$ that occur when agents in $\coalition$ execute the strategy profile $\profile\strat_\coalition$. 
\end{definition}

\begin{definition}[Semantics \ATLD]
	\label{def:ATLD-semantics}
	Let 
    $\CGS = (\setpos, \val, \availability, \transCGS)$ be a CGS, and $\pos \in \setpos$ be a state.
The \emph{semantics} of \ATLD is defined by induction as follows (we omit Boolean cases for brevity): 
\begingroup
\allowdisplaybreaks
    \begin{alignat*}{3}
     & \CGS, \pos  \models  \coop{\coalition} \X \varphi  && \text{ iff } && 
    \exists \profile\strat_\coalition \in \setstrat_\coalition,  
    \forall \iplay \in \out(v,\profile\strat_{\coalition}): 
    \\  & && && 
        \CGS, \iplay_{1} \models \varphi 
     \\ 
     & \CGS, \pos  \models  \coop{\coalition}  \varphi \until \psi  && \text{ iff } && 
    \exists \profile\strat_\coalition \in \setstrat_\coalition,  
    \forall \iplay \in \out(v,\profile\strat_{\coalition}):  
    \\  & && && 
    \exists {i \geqslant 0}  \text{ s.t.  } 
     \CGS, \iplay_{i} \models \psi \text{ and }  
     \\  & && && 
     \forall {0 \leqslant j <i}, \CGS, \iplay_{j} \models \varphi 
     \\
     &  \CGS, \pos  \models  \coop{\coalition}  \varphi \release \psi  && \text{ iff } && 
    \exists \profile\strat_\coalition \in \setstrat_\coalition,  
    \forall \iplay \in \out(v,\profile\strat_{\coalition}):  
    \\  & && && 
    \forall k \geqslant 0, \CGS,\iplay_{k} \models \psi \text{ or } 
    \\  & && && 
    \exists  j \in [0,k)  \text{  s.t. }  \CGS,\iplay_{j}\models \varphi
    \\
      & \CGS, \pos  \models [\overline{\varphi: \act \to \coalition}]^+_k \psi && \text{ iff } &&
    (\setpos, \val, \availability',\transCGS), \pos \models \psi, \text{ where }  
    \\  & && &&
    \forall \pos'\in \setpos, \ag \in \Ag, d'(a,v') =
    \\  & && &&  = d(a,v') \cup  upd^{\pi}(\ag, \pos') 
    \\
         &  	
       \CGS, \pos  \models [\overline{\varphi: \act \to \coalition}]^-_k \psi && \text{ iff } &&
    (\setpos, \val, \availability',\transCGS), \pos \models \psi, \text{ where }  
    \\  & && &&
    \forall \pos'\in \setpos, \ag \in \Ag, d'(a,v') =   \\  & && && = d(a,v') \setminus upd^\pi(\ag, \pos') 
    \\  & && && \text{if } d(a,v') \not\subseteq upd^\pi(\ag, \pos'), 
    \text{ and } 
    \\  & && &&d'(a,v') = d(a,v') \text{ otherwise} 
    \\
    & \CGS, \pos  \models \hasAction{\ag}{\act} && \text{ iff } &&  \act \in \availability(\ag,\pos)
	\end{alignat*}
    \endgroup
where $upd^\pi (\ag, \pos') = \{\alpha_i \mid   1 \leqslant i \leqslant k, \CGS, \pos'  \models \varphi_i, \text{ and } a \in A_i\}$ is the set of actions to be added to or removed from the set of available actions of agent $\ag$ in state $\pos'$ according to the given update $\pi = \varphi_1: \act_1 \to \coalition_1, \dots, \varphi_k: \act_k \to \coalition_k$. We will denote the new \CGS $(\setpos, \val, \availability',\transCGS)$ as $\CGS'$.
\end{definition}

Given an update $[\varphi_1: \act_1 \to \coalition_1, \dots, \varphi_k: \act_k \to \coalition_k ]^+$, we construct a new \CGS $(\setpos, \val, \availability',\transCGS)$, where for each $1 \leqslant i \leqslant k$, we globally and simultaneously grant agents in $A_i$ the action $\act_i$ in states satisfying $\varphi_i$. For the action removing operator $[\varphi_1: \act_1 \to \coalition_1, \dots, \varphi_k: \act_k \to \coalition_k ]^-$ we need to additionally check whether the simultaneous and global removal of actions from some agent $\ag$ is possible, i.e., whether after such an update agent $\ag$ still has at least one action available. If yes, we proceed with updating the set of available actions of $\ag$, and if not, the agent retains \textit{all} her actions\footnote{Note that we could also have defined the removing operator in such a way that if a removal is impossible for some agent $\ag$, we still remove \textit{some} of her actions. In this case, however, we would need to specify an ordering over $\ag$'s actions to decide which ones to remove. This variation is interesting for future work, but here, for simplicity, we do not distinguish between actions of $\ag$. 
}. The assumption that, after an update, an agent should have at least one action available is called \textit{reasonableness} and is standard in reasoning about normative MAS \cite{AgotnesHRSW07}. 

\begin{example}[Oxygen tank]
\label{ex:r1} 
We draw inspiration from the classic example from \cite{bulling2010model}, and adapt it to our context. 
Imagine that two robots, $r_1$ and $r_2$, are transporting an oxygen tank in a circular environment. Their task is to deliver the tank safely to Bob who has a respiratory disorder. 
Robot $r_1$ can move to the left and right, but robot $r_2$ can only go left, 
as she cannot turn. 
Since the tank is heavy, they only succeed to bring the tank when moving in the same direction. The environment has three areas: Bob's room, the salon, and a balcony. 
The model, denoted $\CGS^{Bob}$, is depicted in Figure \ref{fig:robot1} (left). 
State $q_1$ represents the robots successfully bringing the oxygen tank to Bob’s room, making the proposition $atBob$ true. 
States $q_0$ and $q_2$ correspond to the case where the robots instead bring the tank to the salon (their initial position)
and the balcony, respectively. Since the temperature outside is warm, the proposition $warm$ is true at $q_2$. 
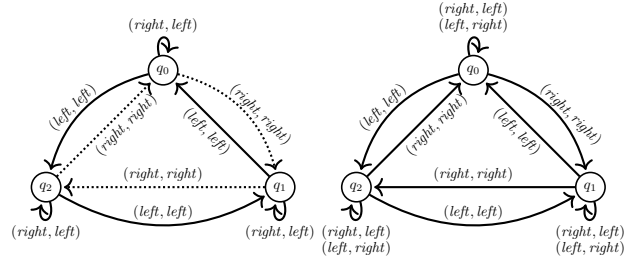
\begin{figure}
    \centering
\begin{tikzpicture}[->,shorten >=1pt,auto,node distance=4cm, semithick,scale=0.55, transform shape]

\begin{scope}[shift={(0,0)}]

    \node[circle, draw, fill=white] (A) {$q_0$};
    \node[circle, draw, fill=white] (B) [below right of=A] {$q_1$};
    \node[circle, draw, fill=white] (C) [below left of=A] {$q_2$};

    \path [->,thick,  bend right] (A) edge node[swap,pos=.3,rotate=45] {$(\leftaction,\leftaction)$} (C);
    \path [->,thick] (B) edge node[pos=.3,rotate=-45] {$(\leftaction,\leftaction)$} (A);
    \path [->,thick, bend right] (C) edge node[pos=.5] {$(\leftaction,\leftaction)$} (B);
    
    \path [->,thick,bend left,densely dotted] (A) edge node[pos=.3,rotate=-45] {$(\rightaction,\rightaction)$} (B); 
    \path [->,thick,densely dotted] (C) edge node[swap,pos=.3,rotate=45] {$(\rightaction,\rightaction)$} (A); 
    \path [->,thick,densely dotted] (B) edge node[pos=.5,swap] {$(\rightaction,\rightaction)$} (C); 
        
    \path [->,thick] (A) edge[loop above] node[pos=.5] {$(\rightaction,\leftaction)$
    } (A);
    \path [->,thick] (B) edge[loop below] node[pos=.5] {$(\rightaction,\leftaction)$} (B);
    \path [->,thick] (C) edge[loop below] node[pos=.5] {$(\rightaction,\leftaction)$} (C);

\end{scope}

\begin{scope}[shift={(7.5cm,0)}]  

    \node[circle, draw, fill=white] (A2) {$q_0$};
    \node[circle, draw, fill=white] (B2) [below right of=A2] {$q_1$};
    \node[circle, draw, fill=white] (C2) [below left of=A2] {$q_2$};

    \path [->,thick,bend left] (A2) edge node[pos=.3,rotate=-45] {$(\rightaction,\rightaction)$} (B2); 
    \path [->,thick,bend right] (A2) edge node[swap,pos=.3,rotate=45] {$(\leftaction,\leftaction)$} (C2); 
    \path [->,thick] (C2) edge node[swap,pos=.3,rotate=45] {$(\rightaction,\rightaction)$} (A2);               
    \path [->,thick] (B2) edge node[pos=.3,rotate=-45] {$(\leftaction,\leftaction)$} (A2); 
    \path [->,thick] (B2) edge node[pos=.5,swap] {$(\rightaction,\rightaction)$} (C2); 
    \path [->,thick,bend right] (C2) edge node[pos=.5] {$(\leftaction,\leftaction)$} (B2);          
    
    \path [->,thick] (A2) edge[loop above] node[align=left,pos=.5] {$(\rightaction,\leftaction)$ \\ $(\leftaction,\rightaction)$
    } (A2); 
    \path [->,thick] (B2) edge[loop below] node[align=left,pos=.5] {$(\rightaction,\leftaction)$ \\ $(\leftaction,\rightaction)$
    } (B2); 
    \path [->,thick] (C2) edge[loop below] node[align=left,pos=.5] {$(\rightaction,\leftaction)$ \\ $(\leftaction,\rightaction)$
    } (C2);
\end{scope}

\end{tikzpicture}

    \caption{The CGS $\CGS^{Bob}$ (left), and its update in which the robot $r_2$ was granted the new action $right$ (right).  Dotted arrows represent disabled transitions, which could be enabled if the robot is granted a new action.} 
    \label{fig:robot1}
\end{figure}
\end{example}

The robots can 
cooperate to bring the oxygen tank to Bob by moving together to the left twice, passing through the balcony.  Thus, we have that $\CGS^{Bob}, q_0 \models \coop{\{r_1,r_2\}} \F \,  atBob$. However, the robots cannot immediately bring it to Bob  from $q_0$ (i.e., in one step). By granting the action $\rightaction$ to $r_2$, which she currently lacks, they can cooperate to bring the tank anywhere in one step (see Figure \ref{fig:robot1}  (right)). Hence,
\begin{align*}
&\CGS^{Bob},\, q \models
    \lnot has(r_2, \rightaction) \, \land \\
    &[\top: \rightaction \to \{r_2\}]^+(has(r_2, \rightaction) \land \coop{\{r_1,r_2\}} \X \, atBob) 
\end{align*}
for any state $q$.  In other words, initially $r_2$ does not have action $\rightaction$, and after she acquires it universally, the robots always have a strategy to bring the tank to Bob in one step. 

It is advised to not expose the oxygen tank to warm temperatures and direct sunlight. Hence, 
we would then like to prevent the robots from going to the balcony before reaching Bob's room. This could be achieved by removing actions leading to the undesirable state. For instance, we have that 
\begin{align*}
\CGS^{Bob},\, q_0 \models
    [\top : \rightaction \to \{r_2\}]^+\,
    [\neg atBob : \leftaction \to \{r_1\}]^- \\
    \coop{\emptyset}\, atBob \, \release \,\neg warm 
\end{align*}

This means that, after the update, whatever the robots do, they are only able to go to the warm area after bringing the oxygen tank to Bob's room. 
The model obtained with the updates as specified in this formula is shown in Figure \ref{fig:robot2}.

\begin{figure}[h]
    \centering
\begin{tikzpicture}[->,shorten >=1pt,auto,node distance=4cm, semithick,scale=0.55, transform shape]

\begin{scope}[shift={(0,0)}]

    \node[circle, draw, fill=white] (A) {$q_0$};
    \node[circle, draw, fill=white] (B) [below right of=A] {$q_1$};
    \node[circle, draw, fill=white] (C) [below left of=A] {$q_2$};

    \path [->,thick,bend left] (A) edge node[pos=.3,rotate=-45] {$(\rightaction,\rightaction)$} (B); 
   
    \path [->,thick] (C) edge node[swap,pos=.3,rotate=45] {$(\rightaction,\rightaction)$} (A);               
    \path [->,thick] (B) edge node[pos=.3,rotate=-45] {$(\leftaction,\leftaction)$} (A); 
    \path [->,thick] (B) edge node[pos=.5,swap] {$(\rightaction,\rightaction)$} (C);

    \path [->,thick,bend right,densely dotted] (A) edge node[swap,pos=.3,rotate=45] {$(\leftaction,\leftaction)$} (C); 
    \path [->,thick,bend right,densely dotted] (C) edge node[pos=.5] {$(\leftaction,\leftaction)$} (B);          

    \path [->,thick] (A) edge[loop above] node[align=left,pos=.5] {$(\rightaction,\leftaction)$ 
    } (A); 
    \path [->,thick] (B) edge[loop below] node[align=left,pos=.5] {$(\rightaction,\leftaction)$ \\ $(\leftaction,\rightaction)$
    } (B); 
    \path [->,thick] (C) edge[loop below] node[align=left,pos=.5] {$(\rightaction,\leftaction)$ 
    } (C);
\end{scope}

\end{tikzpicture}

    \caption{Update of the CGS $\CGS^{Bob}$ obtained by granting action $right$ to $r_2$ and removing action $\leftaction$ of $r_1$. }
    \label{fig:robot2}
\end{figure}
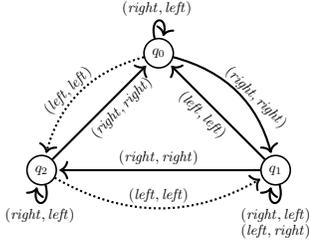



\subsection{Expressivity}

\label{sec:expressivity}

\subsubsection{Expressivity of \ATLD and its Fragments}

In this section, we show that \ATLD is strictly more expressive than \ATL, and we also make the relation between $\ATLD^+$ and $\ATLD^-$ precise.

\begin{definition}[Expressivity]
Let $\mathsf{L}_1$ and $\mathsf{L}_2$ be two languages, and let $\varphi \in \mathsf{L}_1$ and $\psi \in \mathsf{L}_2$. 
We say that $\varphi$ and $\psi$ are \emph{equivalent}, when for all CGS $\CGS, \pos$: $\CGS, \pos \models \varphi$ if and only if $\CGS, \pos \models \psi$.

If for every $\varphi \in \mathsf{L}_1$ there is an equivalent $\psi \in \mathsf{L}_2$, we write $\mathsf{L}_1 \preccurlyeq \mathsf{L}_2$ and say that $\mathsf{L}_2$ is \emph{at least as expressive as} $\mathsf{L}_1$. We write $\mathsf{L}_1 \prec \mathsf{L}_2$ iff $\mathsf{L}_1 \preccurlyeq \mathsf{L}_2$ and $\mathsf{L}_2 \not \preccurlyeq \mathsf{L}_1$, and we say that $\mathsf{L}_2$ is \emph{strictly more expressive than} $\mathsf{L}_1$. Finally, if $\mathsf{L}_1 \not \preccurlyeq \mathsf{L}_2$ and $\mathsf{L}_2 \not \preccurlyeq \mathsf{L}_1$, we say that $\mathsf{L}_1$ and $\mathsf{L}_2$ are \emph{incomparable}.
\end{definition}

\begin{theorem}
    \ATL $\prec$ $\ATLD^+$ and $\ATLD^+ \not \preccurlyeq \ATLD^-$.
\end{theorem}

\begin{proof}
    Consider formula $[\top:\beta \to \{a\}]^+ \coop{\{a\}}\X \lnot p$ of $\ATLD^+$ and assume, for the sake of a contradiction, that there are equivalent formulas $\psi$ of \ATL and $\chi$ of $\ATLD^-$. Next, consider CGSs $\CGS_1$ and $\CGS_2$ with a single agent $a$ and two actions $\{\alpha, \beta\}$ shown in Figure \ref{fig:proof1}.

   \begin{figure}[h!]
\centering
   \begin{tikzpicture}[scale=0.7, transform shape]
\node[circle,draw=black](s) [label=below:$p$] at (0,0) {$v$};

\draw [<-,thick](s) to [loop above, out=45, in=135, looseness = 5] node[above] {$\alpha,\beta$} (s); 

\node[circle,draw=black](s) [label=below:$p$] at (2,0) {$v$};
\node[circle,draw=black](t) at (4,0) {$w$};

\draw[->,thick,bend right, densely dotted] (s) to node[below] {$\beta$}  (t);
\draw[->,thick,bend right, densely dotted] (t) to node[above] {$\beta$} (s);
\draw [<-,thick](s) to [loop above, out=45, in=135, looseness = 5] node[above] {$\alpha$} (s); 
\draw [<-,thick](t) to [loop above, out=45, in=135, looseness = 5] node[above] {$\alpha$} (t);

\node[circle,draw=black](s) [label=below:$p$]at (6,0) {$v$};
\node[circle,draw=black](t) at (8,0) {$w$};

\draw[->,thick,bend right] (s) to node[below] {$\beta$}  (t);
\draw[->,thick,bend right] (t) to node[above] {$\beta$} (s);
\draw [<-,thick](s) to [loop above, out=45, in=135, looseness = 5] node[above] {$\alpha$} (s); 
\draw [<-,thick](t) to [loop above, out=45, in=135, looseness = 5] node[above] {$\alpha$} (t);

\end{tikzpicture}

\caption{CGS $\CGS_1$ (left), $\CGS_2$ (middle), and $\CGS_3$ (right). Dotted arrows denote disabled transitions.}
\label{fig:proof1}
\end{figure}
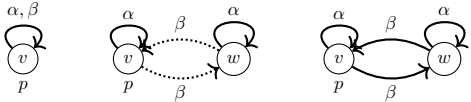 

CGS $\CGS_1$ consists of a single state $v$ satisfying $p$ and with two self-loops for both actions $\alpha$ and $\beta$. In CGS $\CGS_2$, we have two states, $v$ and $w$, with $p$ being true only in $v$. Note that agent $a$ does not have action $\beta$ available to her in either state, i.e. $\beta$-transitions are disabled for her.

It is easy to see that $\CGS_1, v \not \models [\top:\beta \to \{a\}]^+ \coop{\{a\}}\X \lnot p$, since agent $a$ already has action $\beta$ available to her, and no matter what she does, she cannot reach a $\lnot p$-state as there is simply no such state in the CGS. At the same time,  $\CGS_2, v \models [\top:\beta \to \{a\}]^+ \coop{\{a\}}\X \lnot p$, as enabling action $\beta$ for agent $a$ in both states would results in $\CGS_3,v$, where $\coop{\{a\}}\X \lnot p$ holds. 

Finally, $\CGS_1, v \models \psi$ iff $\CGS_2, v \models \psi$ due to the fact that no \ATL formula can create a transition between the $p$-state $v$ and $\lnot p$-state $w$ in $\CGS_2$. Hence, $\ATLD^+ \not \preccurlyeq \ATL$, and, together with the fact that $\ATL \subseteq \ATLD^+$, we get \ATL $\prec$ $\ATLD^+$.

Similarly, $\CGS_1, v \models \chi$ iff $\CGS_2, v \models \chi$ as the only non-trivial removal is possible in $\CGS_1, v$: either action $\alpha$ or action $\beta$ can be removed. In both outcomes, further non-trivial removals are not possible, and hence in both $\CGS_1, v$ and $\CGS_2, v$ the agent $a$ can only force $p$. Thus, $\ATLD^+ \not \preccurlyeq \ATLD^-$.
\end{proof}

\begin{restatable}{theorem}{thmtwo}
    \ATL $\prec$ $\ATLD^-$ and $\ATLD^- \not \preccurlyeq \ATLD^+$.
\end{restatable}

The theorem can be proved by taking $\CGS_3$ depicted in Figure~\ref{fig:proof1}, and $\CGS_4$, which is exactly like $\CGS_3$ with actions $\alpha$ and $\beta$ swapped. Then we can target action $\beta$ with an action removing update and show that the corresponding formula holds in one model and not in the other one. A proof sketch can be found in the Technical Appendix.

\begin{corollary}
\label{cor:exp}
    $\ATLD^+$ and $\ATLD^-$ are incomparable, $\ATLD^+ \prec \ATLD$, and $\ATLD^- \prec \ATLD$.
\end{corollary}

The overview of the expressivity results is presented in Figure \ref{fig:landscape}.

   \begin{figure}[h!]
\centering
   \begin{tikzpicture}[scale=0.8, transform shape]
\node (s) at (0,0) {\ATL};

\node(t) at (2,1) {$\ATLD^+$};
\node(u) at (2,-1) {$\ATLD^-$};
\node(v) at (4,0) {$\ATLD$};

\draw[->,thick] (s) to  (t);
\draw[->,thick] (s) to  (u);
\draw[->,thick] (t) to  (v);
\draw[->,thick] (u) to  (v);

\end{tikzpicture}
 
\caption{Overview of the expressivity results. An arrow from $\mathsf{L}_1$
to $\mathsf{L}_2$ means $\mathsf{L}_1 \prec \mathsf{L}_2$. 
The relation is transitive,  and transitive arrows are omitted from the figure.}
\label{fig:landscape}
\end{figure}
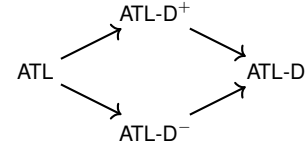 

\subsubsection{Relation to Logics for Normative Systems}

In this paper, we focus on a general approach to granting and removing agents' actions. Action removals have been explored before in the context of \textit{normative} MAS, in which a \textit{norm} or a \textit{social law} prohibits certain transitions or actions (see \cite{AlechinaLD18} for an overview). Here we show how $\ATLD^-$ captures some intuitions from the literature on the topic. 

Let us first consider the approach from \cite{HoekRW07}, in which, in the \ATL setting, the authors consider, for a given \CGS $\CGS$, \textit{behavioural constraints} $\eta: \Act \times \Ag \to 2^{V}$. If $v \in \eta (\alpha,a)$, then it means that agent $a$ is not allowed to perform action $\alpha$ in state $v$. The \textit{implementation of $\eta$ in $\CGS$} is defined as the new \CGS $\CGS \dagger \eta$, in which all transitions that include prohibited actions are removed. Note that, as in our case, behavioural constraints must satisfy the reasonableness requirement. The pair $(\varphi, \eta)$ is called a \textit{social law}, where  $\eta$ is a behavioural constraint and $\varphi$ is a goal formula. 

Since behavioural constraints explicitly specify states where an action is prohibited, we cannot directly model them in $\ATLD^-$. The reason for this is that our updates are not tailored for a specific model, and instead we use conditions expressed as formulas to pick out the necessary states. Indeed, the approach from \cite{HoekRW07} allows, for example, removing action $\alpha$ from agent $a$ in state $v$, while preserving the action in state $v'$, even if $v$ and $v'$ are not distinguishable by any formula. We do not delve into the discussion whether such a situation is desirable, and instead note that, given a finite $\CGS$ and a goal formula $\varphi$, we can label each state of $\CGS$ with its own unique proposition that is not true in any other state of $\CGS$ and does not appear in $\varphi$. We denote such a \CGS $\CGS^+$, and it is clear that $\CGS, v \models \varphi$ iff $\CGS^+, v \models \varphi$ for all $v$. 

This modification allows us to target 
any particular state of a \CGS, and hence model the effects of implementing a behavioural constraint $\eta$. In particular, we create an update $[\pi_\eta]^-$, such that if $v \in \eta (\alpha,a)$ for some $v, \alpha,$ and $a$, then $\pi_\eta := \pi_\eta, p_v: \alpha \to \{a\}$, where $p_v$ is a new proposition that is true only in $v$ and not appearing in $\varphi$. Hence, we have the following result.

\begin{theorem}
\label{thm:norm1}
    Let $(\varphi, \eta)$ be a social law with $\varphi \in \ATL$, and $\CGS$ be a finite \CGS. The social law is effective, i.e. $\CGS\dagger \eta, v \models \coop{\emptyset} \G \varphi$, iff $\CGS^+, v \models [\pi_\eta]^- \coop{\emptyset} \G \varphi$. 
\end{theorem}

Now, let us consider the approach from \cite{GalimullinK24}, in which the authors, again in the \ATL setting, consider partial functions $\zeta: \Act \times \Ag \hookrightarrow \mathsf{L}$, called \textit{atomic social laws}, specifying conditions under which an action is allowed to be performed by an agent. The main difference from behavioural constraints is that atomic social laws specify the states in which an action is allowed using formulas of a language. For example, $\zeta(\alpha, a) = \varphi$ means that agent $a$ is allowed to perform action $\alpha$ only in states satisfying $\varphi$. If for some $a$ and $\alpha$, $\zeta(\alpha, a)$ is undefined, the action is unconditionally allowed for the agent. Having a goal formula $\varphi$ and an atomic social law $\zeta$, we denote the fact that $\varphi$ holds after updating the \CGS $\CGS$ with $\zeta$ as $\CGS, v \models [\zeta]\varphi$.

Atomic social laws are closer to the updates of $\ATLD^-$ than behavioural constraints. For an atomic social law $\zeta$, we construct $[\pi_\zeta]^-$ as follows: for all $a \in \Ag$ and $\alpha \in \Act$, if $\zeta(\alpha, a) = \varphi$, then $\pi_\zeta:= \pi_\zeta, \lnot \varphi: \alpha \to  \{a\}$, i.e. if $\alpha$ is only allowed in $\varphi$-states, we remove this action from agent $a$ in all $\lnot \varphi$-states.  Thus, we have the following result. 

\begin{theorem}
\label{thm:norm2}
    Given are a \CGS $\CGS$, a state $v$, a goal formula $\varphi \in \ATL$, and an atomic social law $\zeta$. It holds that $\CGS,v \models [\zeta] \varphi$ iff $\CGS, v \models [\pi_\zeta]^- \varphi$.
\end{theorem}

Theorems \ref{thm:norm1} and \ref{thm:norm2} illustrate how $\ATLD^-$ can capture some key intuitions from normative reasoning research. The full exploration of the relationship between \ATLD and logics for normative reasoning in MAS is left for future work. 

\subsection{Computational Problems}
\label{sec:modelchecking}
\subsubsection{Model Checking}
Now we turn our attention to the computational problems, 
and we start with model checking.

\begin{definition}[Model checking]
    Given a finite CGS $\CGS =(\setpos, \val, \availability, \transCGS)$ 
    and a formula $\varphi \in \ATLD$, the \emph{global model checking problem} consists in computing $\{\pos \in \setpos \mid \CGS,\pos \models \varphi\}$, i.e. finding all the states of $\CGS$ that satisfy $\varphi$.
\end{definition}

To show that model checking \ATLD is in \PTIME, we present Algorithm \ref{MC} below, omitting the standard \ATL cases.

\begin{breakablealgorithm}
	\caption{Global model checking for \ATLD}\label{MC} 
 \footnotesize
	\begin{algorithmic}[1] 		
		\Procedure{MC}{$\CGS,\varphi$}		
        \Case {$\varphi = \hasAction{\ag}{\act}$}
            \State{\textbf{return} $\{\pos \in V|\act \in \availability(\ag,\pos)\}$}
        \EndCase
        
        \Case {$\varphi = [\pi]^+ \chi$ with $[\pi]^+=  [\overline{\psi: \act \to A} ]^+_k $}
        \For{ $v' \in V$ and $a \in \Ag$}
        \State{$X(a,\pos') = \emptyset$}
        \For {$i \in [1,...,k]$}
        \If{$\pos' \in \textsc{MC}(\CGS, \psi_i)  \text{ and } \ag \in A_i$}
        \State{$X(a,\pos') \gets X(a,\pos') \cup \{\alpha_i\}$}
        \EndIf
        \EndFor
        \State{$\availability{'}(\ag,\pos') \gets \availability(\ag,\pos') \cup X(a,\pos')$}
        \EndFor
        \State{\textbf{return}  $\textsc{MC}  (\CGS',  \chi)$}
        \EndCase
                \Case {$\varphi = [\pi]^-\chi$ with $[\pi]^-= [\overline{\psi: \act \to A} ]^-_k$}
        \For{ $v' \in V$ and $a \in \Ag$}
         \State{$X(a,\pos') = \emptyset$}
         \For {$i \in [1,...,k]$}
        \If{$\pos' \in \textsc{MC}(\CGS, \psi_i) \text{ and } \ag \in A_i$}
        \State{$X(a,\pos') \gets X(a,\pos') \cup \{\alpha_i\}$}
        \EndIf
        \EndFor
        \If{$\availability(\ag,\pos') \not \subseteq X(a,\pos')$}
        \State{$\availability'(\ag,\pos') \gets \availability(\ag,\pos') \setminus X(a,\pos') $}
        
        \Else
        \State{$\availability'(\ag,\pos') \gets \availability(\ag,\pos')$}
        \EndIf
        \EndFor
        \State{\textbf{return}  $\textsc{MC}  (\CGS', \chi)$}
        \EndCase
   \EndProcedure

	\end{algorithmic}
\end{breakablealgorithm}

The algorithm is an extension of the \ATL global model-checking algorithm, which is computable in polynomial time. The original \ATL algorithm employs the function $Pre(\CGS, A, X)$ to compute the set of states from which the coalition $A$ can force an outcome in one of the $X$ states. Function $Pre$ is computable in polynomial time w.r.t. $|\CGS|$. In our case, the difference is that we may also compute $Pre$ not only for the original model but for at most $|\varphi|$ updated models as well. Updated models are computed for cases $[\pi]^+\chi$ and $[\pi]^-\chi$, where we, following the definition of semantics, globally update actions of agents in a given CGS $\CGS$ by iteratively checking the lists $[\overline{\psi: \act \to A} ]^+_k$ and $[\overline{\psi: \act \to A} ]^-_k$. 

To update models, we recursively call the procedure \textsc{MC}($\CGS,\varphi$) at most $|V|^2\cdot|\Ag|\cdot k\cdot |\varphi|$ times. As mentioned, each call takes at most polynomial time, and hence the whole algorithm terminates in polynomial time w.r.t. $|\CGS|$ and $|\varphi|$.  

The \PTIME-hardness of the problem follows immediately from the \PTIME-completeness of model checking for \ATL (a fragment of \ATLD).

\begin{theorem}
    Model checking \ATLD is \PTIME-complete.
\end{theorem}

\subsubsection{Bounded Update Existence Problem}

Next, we consider the bounded update existence problem, which asks whether, for a given \CGS and a goal formula, there is a bounded sequence of updates that modifies the \CGS so that it satisfies the goal. We show that this problem is \textsc{NP}-complete.

\begin{definition}[Bounded Update Existence Problem]
Given a finite CGS $\CGS$, a state $\pos$, formula $\varphi$, and a natural number $n$, the \textit{bounded update existence problem} consists in determining whether there is a sequence $[\overrightarrow{\pi}]$ of updates $[\pi]^+$ and $[\pi]^-$ with $|\overrightarrow{\pi}| \leqslant n$ such that $\CGS, \pos \models [\overrightarrow{\pi}]\varphi$.
\end{definition}

Considering the bounded variant of the update existence problem is important when we want \textit{economical modifications} of a given system. Indeed, dramatic changes to a system are not always desirable or possible due to resource constraints. In these cases, parameter $n$ may represent the upper bound on how much we are willing to change the system.   

\begin{restatable}{theorem}{synthesis}
The bounded update existence problem is \textsc{NP}-complete.
\end{restatable}

\begin{proof}[Proof sketch]
We show that the problem is \textsc{NP}-complete for all of $\ATLD$, $\ATLD^-$, and $\ATLD^+$.
First, that the problem is in \textsc{NP} for all the logics  follows from the fact that for given $\CGS$, $\pos$, and $\varphi$, we can guess the sequence $[\overrightarrow{\pi}]$ such that $|[\overrightarrow{\pi}]\varphi|$ is polynomial with respect to $n+|\varphi| + |\CGS|$, and model check $\CGS, \pos \models [\overrightarrow{\pi}]\varphi$ in polynomial time.

To demonstrate that the problem is \textsc{NP}-hard for $\ATLD^+$, and hence $\ATLD$, we employ a reduction from the \textsc{NP}-complete 3SAT problem. For a given 3SAT instance $\Psi$, we create a CGS $\CGS_\Psi$ for one agent, where states correspond to clauses in $\Psi$, literals, atoms, and, finally, the truth-values of atoms. Then, for thus constructed $\CGS_\Psi$, we provide an \ATL formula $\varphi_{\Psi}$ that forces the agent to unambiguously choose truth-values for atoms to satisfy the goal $\Psi$. We claim that $\Psi$ is satisfiable if and only if there is an update $[\pi]^+$ such that once it is executed in $\CGS_\psi$, then $\varphi_\Psi$ holds. Such an update would choose which truth-values to assign to atoms in $\Psi$. See the detailed proof in the Technical Appendix.

The \textsc{NP}-hardness for $\ATLD^-$ can be obtained by a modification of the proof for $\ATLD^+$, where we consider the same $\varphi_{\Psi}$, but in $\CGS_\Psi$ all actions are enabled. Thus, $\Psi$ is satisfiable if and only if there is an update $[\pi]^-$ that disables some actions for the agent allowing for an unambiguous assignment of truth-values to atoms to satisfy $\varphi_{\Psi}$. 
\end{proof}

The \textsc{NP}-completeness of the bounded update existence problem for \ATLD echoes the \textsc{NP}-completeness result for (both bounded and unbounded, i.e. without a given $n$) social law existence problem \cite{GalimullinK24}. 


\section{Epistemic ATL with Dynamic Actions}\label{sec:ATELD}



While in \ATLD we can reason about granting and removing actions to affect the agents' abilities, in order to increase or decrease the autonomy of the agents, we would like them to know facts about the environment, each other, and which actions they need to acquire or drop to achieve their goals. For example, knowing what actions they lack, the agents then could  \emph{proactively} request to acquire this action. To this end, we introduce an \textit{epistemic} extension of \ATLD.  

\subsection{Syntax and Semantics}

\textit{Epistemic} \ATLD (\ATELD) is a dynamic extension of \ATEL \cite{HoekW03a,JamrogaH04}, which is defined over CGSs enriched with relations specifying agents' indistinguishable states.


\begin{definition}[\iCGS]
  \label{def-icgs}
A \emph{ concurrent game structure with imperfect information} (\iCGS) is a tuple
$\iCGS=(\CGS, \{\obsrel\}_{\ag\in\Ag})$ where $\CGS=(\setpos, \val, \availability, \transCGS)$ is a CGS, and $\obsrel\;\subseteq \setpos\times\setpos$ is an equivalence relation (called the  \emph{observation relation}) of agent $\ag$ s.t. if $\pos \sim_a \pos'$, then $d(a,v) = d(a,v')$. We write $\iCGS=(\CGS, \sim)$ whenever the relation is clear from the context. 
\end{definition}

Note that a standard \CGS is a special case of an \iCGS, where for each agent $a$ the relation $\sim_a$ is the identity relation. The requirement that agents have the same actions in indistinguishable states is standard in the \ATEL literature\footnote{For the discussion of this property see \cite{Agotnes06}.}. 

\begin{definition}[Uniform Strategies]
    A \emph{memoryless uniform  strategy} for agent $\ag$
is a function $\sigma_\ag: \setpos \to \Act$ 
in which for all states $\pos,\pos'$ s.t. $\pos\obsrel\pos'$, we have $\strat_\ag(\pos)=\strat_\ag(\pos')$.  With a slight abuse of notation, we let $\setstrata$ be the set of memoryless uniform strategies for agent $\ag$, and $\setstrat_\coalition=\prod_{\ag \in \coalition}\setstrat_{\ag}$. 
%
\end{definition}

Intuitively, uniform strategies assign the same actions in indistinguishable states. 
The assumption that agents know actions available to them has epistemic side effects in our setting. As we will see, granting a new action with a condition $\varphi$ to a set of agents $A$, also lets the agents distinguish $\varphi$-states from $\lnot \varphi$-states. This is somewhat similar to \textit{semi-private announcements} in  \textit{Dynamic Epistemic Logic} \textsf{DEL}  \cite{hvdetal.del:2007}, where a set of agents $A$ may be informed whether $\varphi$ holds, and all other agents know that $A$ were informed whether $\varphi$ but do not know themselves what exactly was communicated. 

Here, instead, we consider a more general setting in which an update can also inform a (possibly empty) subset $B$ of $\Ag$ of the action update for agents in $A$ 
in $\varphi$-states. This general setting subsumes semi-private announcements by taking $B = \emptyset$. At the same time, it allows us to express \textit{public} granting and removal of actions by taking $B = \Ag$, where \textit{all} agents in a system are informed of the condition of an update.  
As an example, we can think of chess: when one player gets her pawn promoted, the other player  knows it. Similarly, if the player loses a piece, it is common knowledge 
%
that she can no longer use it. 

Finally, with our general setting, we can express situations `in-between' semi-private and public action updates. This allows us to capture scenarios, where 
%
we would like to inform only a subset $B$ of $\Ag$ that group $A$ has received or lost an action. We can think of such groups $A$  and $B$ as teammates in a board game. 
%
 
To capture all of these scenarios, we now present \textit{Epistemic} \ATLD (\ATELD).

\begin{definition}[\ATELD]\label{def:ATELD-syntax}\!\!
 Formulas $\varphi$ of  \ATELD 
are defined by the following grammar:
      \begin{align*}
    	\varphi  ::= &  p  \mid  \varphi \lor \varphi \mid \neg \varphi \mid  
        \coop{\coalition} \X \varphi \mid  
        \coop{\coalition} \varphi \until \varphi \mid  
        \coop{\coalition} \varphi \release \varphi \mid      
        \\  
        &  [\pi]^+\varphi \mid  [\pi]^-\varphi  
         \mid
        \hasAction{\ag}{\act}  \mid  \Ka \varphi    
        \\
        \pi ::= & \; \varphi: \act \to \coalition, B \mid 
        \pi , \pi 
   \end{align*}
	where $p\in\APf$, $\act \in \Act$, $\ag \in \Ag$,  and $\coalition, B\in2^\Ag$.
\end{definition}

Constructs $\Ka \varphi$ are read as `agent $\ag$ knows that $\varphi$'. The new updates $[\varphi: \act \to \coalition, B]^+$ (resp. $[\varphi: \act \to \coalition, B]^-$) mean that in all states satisfying $\varphi$, agents in $\coalition$ are granted action $\act$ (resp. get  action $\act$ removed) and both the agents in $\coalition$ and in $B$ are informed of $\varphi$, i.e., they are informed of the the action update's precondition. 

\begin{definition}[Semantics \ATELD]
	\label{def:ATLED-semantics}
	Let $\varphi$ be an  $\ATELD$ formula, $\iCGS = (\CGS, \sim)$ be an \iCGS and $\pos$ be a state.  
The satisfaction of $\varphi$ on $\pos$ in $\iCGS$ is denoted $\iCGS, \pos  \models \varphi$ and is defined recursively as follows (where we omit the cases that are analogous to \ATLD and the case for removing actions as it is analogous to the case of adding actions): 
\begingroup
\allowdisplaybreaks
    \begin{alignat*}{3}
    &\iCGS, \pos  \models \Ka \varphi && \text{ iff } && \forall \pos' \in \setpos: \pos\obsrel\pos' \text{ implies }\\  & && && \CGS, \pos'  \models \varphi  
    \\
     & \iCGS, \pos  \models [\overline{\varphi: \act \to \coalition,B}]^+_k \psi && \text{ iff } &&
    (\CGS',\sim'), \pos \models \psi,   
    \\  & && &&
    \text{where } \CGS' \text{ as in Def. \ref{def:ATLD-semantics}, and } 
    \\  & && && \forall \ag \in \Ag, \sim_a' = \sim_a \cap \sim_{\varphi_i} 
    \\  & && &&\text{ if } a \in A_i \cup B_i, \text{ and }
    \\  & && && \obsrel[a]' =\obsrel[a] \text{ otherwise}
	\end{alignat*}
    \endgroup
where $\sim_\varphi:= (\llbracket \varphi \rrbracket_{\mathcal{G}_i} \times \llbracket \varphi \rrbracket_{\mathcal{G}_i}) \cup (\llbracket \lnot \varphi \rrbracket_{\mathcal{G}_i} \times \llbracket \lnot \varphi \rrbracket_{\mathcal{G}_i})$ is a \emph{$\varphi$-partition} of the $\iCGS$ with $\llbracket \varphi \rrbracket_{\mathcal{G}_i} = \{\pos \in \setpos \mid \iCGS, \pos \models \varphi\}$.
\end{definition}

Intuitively, an update $[\overline{\varphi: \act \to \coalition,B}]^+_k$ (resp. $[\overline{\varphi: \act \to \coalition,B}]^-_k$) globally and simultaneously grants (resp. removes) actions for agents in $A_i$ in states that satisfy the condition $\varphi_i$. This is exactly as in Definition \ref{def:ATLD-semantics}. Additionally, we inform agents in $A_i \cup B_i$ of the action update's precondition, i.e., the truth value of $\varphi_i$. To model this, we consider the $\varphi_i$-partition\footnote{Such $\varphi$-partitions are inspired by the logical semantics of questions (see \cite{GroenendijkS97} for an overview, \cite{BenthemM12} for the \textsf{DEL} approach to questions, and \cite{diego} for a related discussion). The particular relation $\sim_\varphi$ has been recently used to reason about topic-based communication between agents in \textsf{DEL} \cite{fer,GalimullinV25}.} of the given \iCGS that divides the states of the \iCGS into those that satisfy $\varphi_i$ and those that do not. Then, for each agent $\ag \in A_i \cup B_i$ that should be informed of $\varphi_i$, we `cut' the $\ag$ relations $\obsrel$ that connect $\varphi_i$- and $\lnot \varphi_i$-states. In this way, agents can become less uncertain about which state they are in. 

As mentioned earlier, with the action updates of \ATELD, we can capture a plethora of action update scenarios and their epistemic consequences. For example, the \ATLD updates $[\varphi: \act \to \coalition]^+$ and $[\varphi: \act \to \coalition]^-$  are captured by the \ATELD updates $[\varphi: \act \to \coalition,\emptyset]^+$ and $[\varphi: \act \to \coalition,\emptyset]^-$ correspondingly. Agents in $\coalition$ are still informed about $\varphi$ as the classic assumption of \iCGS is that agents know which actions are available to them, and hence have no uncertainty between the states in which they have different sets of available actions. This is exactly the setting of semi-private action granting or removing. 

Similarly, fully public granting or removal of actions can be captured by $[\varphi: \act \to \coalition,\Ag]^+$ and $[\varphi: \act \to \coalition,\Ag]^-$, where every agent is informed about $\varphi$. Finally, we can also capture \textit{public announcements}\footnote{Our setting differs from the classic public announcements \textit{à la} \cite{Plaza07}, in which a public announcement of $\varphi$ leads to an updated model with all $\lnot \varphi$-states removed. Our case is closer to an alternative variation of public announcements that instead of removing states remove transitions (see more on this in \cite{GerbrandyG97,BenthemL07}).} by using, e.g., $[\varphi: \act \to \emptyset,\Ag]^+$, where agents' actions remain intact and everyone publicly and simultaneously learns the truth-value of $\varphi$.

\begin{example}[Oxygen tank, continued]
Let us 
now assume that  
the robots only have partial observation and may be unsure about their current location. Suppose that robot $r_2$ is
equipped with a temperature sensor and can therefore distinguish the inside areas (Bob's room, i.e. state $q_1$, and the salon, i.e. state $q_0$) from the balcony (state $q_2$), which is warm. 
Robot $r_1$, instead, detects the Wi‑Fi signal strength, which is only strong in Bob's room. 
The model representing this situation is denoted $\CGS^{Bob}_i$ and appears in Figure \ref{fig:robot3} (left).  


\begin{figure}[h]
    \centering
\begin{tikzpicture}[->,shorten >=1pt,auto,node distance=4cm, semithick,scale=0.55, transform shape]

\begin{scope}[shift={(0,0)}]

    \node[circle, draw, fill=white] (A) {$q_0$};
    \node[circle, draw, fill=white] (B) [below right of=A] {$q_1$};
    \node[circle, draw, fill=white] (C) [below left of=A] {$q_2$};

    \path [->,thick,bend left,densely dotted] (A) edge node[pos=.3,rotate=-45] {$(\rightaction,\rightaction)$} (B);  
    \path [->,thick,bend right] (A) edge node[swap,pos=.3,rotate=45] {$(\leftaction,\leftaction)$} (C); 
    \path [->,thick,bend right,densely dotted] (C) edge node[swap,pos=.3,rotate=45] {$(\rightaction,\rightaction)$} (A);
    \path [->,thick,bend left] (B) edge node[pos=.3,rotate=-45] {$(\leftaction,\leftaction)$} (A); 
    \path [->,thick,densely dotted] (B) edge node[pos=.5,swap] {$(\rightaction,\rightaction)$} (C); 
    \path [->,thick,bend right] (C) edge node[pos=.5] {$(\leftaction,\leftaction)$} (B);          
    
    \path [->,thick] (A) edge[loop above] node[align=left,pos=.5] {$(\rightaction,\leftaction)$ 
    } (A); 
    \path [->,thick] (B) edge[loop below] node[align=left,pos=.5] {$(\rightaction,\leftaction)$ 
    } (B); 
    \path [->,thick] (C) edge[loop below] node[align=left,pos=.5] {$(\rightaction,\leftaction)$ 
    } (C);
    \path [-,dashed] (A) edge node[fill=white, anchor=center, pos=0.5]  {$r_1$} (C) ;
    \path [-,dashed] (A) edge node[fill=white, anchor=center, pos=0.5]  {$r_2$} (B) ;
\end{scope}

\begin{scope}[shift={(7.5cm,0)}]  
  \node[circle, draw, fill=white] (A) {$q_0$};
    \node[circle, draw, fill=white] (B) [below right of=A] {$q_1$};
    \node[circle, draw, fill=white] (C) [below left of=A] {$q_2$};

    \path [->,thick,bend left] (A) edge node[pos=.3,rotate=-45] {$(\rightaction,\rightaction)$} (B);  
    \path [->,thick,bend right] (A) edge node[swap,pos=.3,rotate=45] {$(\leftaction,\leftaction)$} (C); 
    \path [->,thick,bend right] (C) edge node[swap,pos=.3,rotate=45] {$(\rightaction,\rightaction)$} (A);
    \path [->,thick,bend left] (B) edge node[pos=.3,rotate=-45] {$(\leftaction,\leftaction)$} (A); 
    \path [->,thick,densely dotted] (B) edge node[pos=.5,swap] {$(\rightaction,\rightaction)$} (C); 
    \path [->,thick,bend right] (C) edge node[pos=.5] {$(\leftaction,\leftaction)$} (B);          
    
    \path [->,thick] (A) edge[loop above] node[align=left,pos=.5] {$(\rightaction,\leftaction)$ \\$(\leftaction,\rightaction)$
    } (A); 
    \path [->,thick] (B) edge[loop below] node[align=left,pos=.5] {$(\rightaction,\leftaction)$ 
    } (B); 
    \path [->,thick] (C) edge[loop below] node[align=left,pos=.5] {$(\rightaction,\leftaction)$ \\ $(\leftaction,\rightaction)$  
    } (C);
    \path [-,dashed] (A) edge node[fill=white, anchor=center, pos=0.5]  {$r_1$} (C) ;

\end{scope}

\end{tikzpicture}

    \caption{The CGS $\CGS^{Bob}_{i}$ with imperfect information is shown on the left, and its update after publicly granting the action $\rightaction$ to $r_2$ whenever the robots are not in Bob's room is shown in the right. 
    Dashed lines labelled by an agent's name represent her observation relation (reflexive cases are omitted). }
    \label{fig:robot3}
\end{figure}
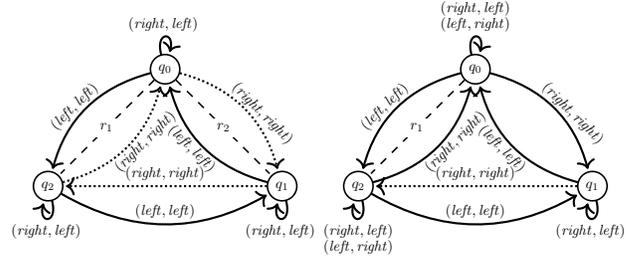
\end{example}

Similarly to Example \ref{ex:r1}, the robots cannot reach Bob's room from the salon in one step, but now they both know it (there is no uniform strategy to force the outcome to satisfy $atBob$ from all states the robots consider indistinguishable), i.e., $\CGS^{Bob}_{i}, q_0 \models  \lnot \coop{\{r_1,r_2\}} \X \, atBob \land K_{r_1} \lnot \coop{\{r_1,r_2\}} \X \, atBob \land K_{r_2} \lnot \coop{\{r_1,r_2\}} \X \, atBob$. We could try to rectify it by the same granting action update as in Example \ref{ex:r1}, while informing both agents about the precondition $\top$.
This update, however, which results in the model similar to the one on the left of Figure \ref{fig:robot3} but with all transitions enabled, is still not enough to guarantee that the robots can reach Bob's room in one step. 
This is due to the fact that their observability relations after the announcement would remain the same ($\top$ holds in all states and hence it's not informative), and they would still be uncertain whether to move to the right or the left (i.e., they have no uniform strategy). Thus,  $\CGS^{Bob}_{i}, q_0 \not\models [\top: \rightaction \to \{r_2\},\Ag]^+ \coop{\{r_1,r_2\}} \X \, atBob$.  

We can now try to grant the action $\rightaction$ to $r_2$ in states other than Bob's room. The resulting model is shown in Figure \ref{fig:robot3} (right). Since $r_2$ is informed about the precondition of granting the action, she now knows that she is in the  salon. However, note that it is still not enough to enable the robots to reach Bob's room in one step as robot $r_1$ is still uncertain whether she should go right (if in state $q_0$) or left (if in state $q_2$). Hence, there is no uniform strategy for the robots to achieve their objective, i.e., 
$\CGS^{Bob}_{i}, q_0 \not \models [\lnot atBob: \rightaction \to \{r_2\},\emptyset]^+ \coop{\{r_1,r_2\}} \, \X \,  atBob$.  To allow the agents to achieve their goal, we can strengthen the precondition to $\lnot atBob \land \lnot warm$ and inform also $r_1$ about this. Since this formula holds only in state $q_0$, the resulting model (which is like in Figure \ref{fig:robot3} (right) but with the observability relation being the identity and transition $(\rightaction,\rightaction)$ from $q_2$ to $q_0$ disabled)  would inform both robots that they are in state $q_0$, from which they have a uniform strategy to reach Bob's room in one step. In particular, $\CGS^{Bob}_{i}, q  \models [\lnot atBob \land \lnot warm: \rightaction \to \{r_2\},\{r_1\}]^+ \coop{\{r_1,r_2\}} \, \X \,  atBob$ holds in all states $q$. This in particular implies that in the initial situation, while both robots know that they cannot reach Bob's room from the salon, they also know that after this particular update, they would be able to achieve their goal, i.e., $\CGS^{Bob}_{i}, q_0 \models K_r \lnot \coop{\{r_1,r_2\}} \, \X \,  atBob \land  K_r[\lnot atBob \land \lnot warm: \rightaction \to \{r_2\},\{r_1\}]^+ \coop{\{r_1,r_2\}} \, \X \,  atBob$ for $r \in \{r_1,r_2\}$. Hence, each of the robots, based on their knowledge, could \textit{proactively} request the system administrator to perform the update in order to allow the robots to reach Bob's room in one step.

\subsection{Model Checking and Expressivity} 

\ATELD extends \ATEL with imperfect information and memoryless strategies (imperfect recall). It is known that model checking \ATL with imperfect information and memoryless strategies, and hence \ATEL, is $\Delta^{\textsc{P}}_2$-complete\footnote{$\Delta^{\textsc{P}}_2 = \textsc{P}^{\textsc{NP}}$ is a class of problems solvable in polynomial time with access to an \textsc{NP} oracle \cite[Chapter 17]{ccomplexity}.} \cite{Schobbens04,JamrogaD06}. The algorithm works in a bottom-up way starting from simple subformulas $\coop{\coalition} \varphi$, where $\varphi$ does not contain strategic modalities. For $\coalition$, the algorithm guesses a uniform strategy of at most polynomial size and removes from the given model transitions that will never be used according to the guessed strategy. Finally, in the resulting model, it is enough to check whether, for all paths, $\varphi$ is satisfied. This can be done by the standard labelling algorithm for \textit{Computation Tree Logic} (\textsf{CTL}) \cite{ClarkeE81}. The number of such guesses of strategies is linear.  

The algorithm needs significant modification for \ATELD, as verfication of a formula $\coop{\coalition} \varphi$ should occur within the updates whose scope contains the formula. In fact, such updates may have changed the available strategies for $\coalition$ by the time we need to verify $\coop{\coalition} \varphi$. So, we must ensure that  the effects of updates are known before the algorithm handles formulas within their scope. To this end, following \cite{kuijer15}, we preprocess the list of subformulas of a given formula. 

Given a $\psi \in \ATELD$, we let $Sub(\psi)$ be the list of subformulas of $\psi$ as well as all modalities $[\pi]^+$ and $[\pi]^-$ appearing in $\psi$. We label all elements of $Sub(\psi)$ by the sequence of updates, possibly empty, in the scope of which a given subformula or update modality appears. Finally, we order the list\footnote{Here we omit the exact details for brevity, but our ordering is quite similar to the ones used in \cite{kuijer15,GalimullinV25}.} by first taking update modalities that are not in the scope of other updates, and ordering their subformulas in the increasing order including the modalities themselves. Then, we proceed with update modalities that are in the scope of the already processed updates. Finally, we take the remaining subformulas and order them in the increasing order. For example, take $\psi = [p: \act \to \{a\},\{b\}]^+[p: \act \to \{b\},\{a\}]^- K_a \coop{\{b\}} \X q$, where we abbreviate the first update as $[\pi]^+$ and the second update as $[\pi]^-$. The ordered list $Sub(\psi)$ is $\{p,$ $[\pi]^+,$ $p^{[\pi]^+},$ ${[\pi]^-}^{[\pi]^+},$ $q^{[\pi]^+, [\pi]^-}, (\coop{\{b\}} \X q)^{[\pi]^+, [\pi]^-}, (K_a \coop{\{b\}} \X q)^{[\pi]^+, [\pi]^-},$ $([\pi]^-K_a \coop{\{b\}} \X q)^{[\pi]^+},$  $\psi\}$.
Such a preprocessing guarantees that a subformula in a scope of an update is evaluated after we know the effects of the update. 


Having prepared the labelled list of subformulas and update modalities, we extend the $\mathsf{CTL}$ labelling algorithm by labelling not only states but actions, transitions, and the observation relation as well. Thus, we label the observation relations $(v, v') 
\in \sim_a$ with a string of updates $\sigma$ to denote that the relation between states $v$ and $v'$ for agent $a$ is still present in the model after the sequence of updates $\sigma$. Similarly, we label actions of each agent at each state with $(\sigma;+)$ if the corresponding action is available to the agent in the corresponding state, and with $(\sigma; -)$ otherwise. We require this labelling to keep track of which actions are available after the sequence of updates $\sigma$.

The model checking algorithm for \ATELD (Algorihm \ref{MC2}) represents the availability function $\availability:\Ag \times V \to 2^\Act \setminus \emptyset$ 
as a set of all possible triples $(a,v,\act)$. 
We label a triple  $(a,v,\act)$ with $(\epsilon; +)$ if $\act \in d(a,v)$, and $(\epsilon; -)$ otherwise, where $\epsilon$ is an empty string. Observe that having all triples explicitly leads to at most a polynomial increase in size.

Now we turn to the algorithm, which  loops over all elements of the ordered labelled list $Sub(\varphi)$ and all states $v$. We write $\pm (a,v, \act)^\sigma$ to refer to the second element of the labelling $(\sigma; \bullet)$ of $(a,v, \act)$, i.e., to either $+$ or $-$. In the algorithm, we focus on updates, and we omit the Boolean 
and $has(\ag,\act)$ cases as they are straightforward.

\begin{breakablealgorithm}
	\caption{Global model checking for \ATELD}\label{MC2} 
 \footnotesize
	\begin{algorithmic}[1] 		
    

        \Case{$\psi^\sigma = (K_a \chi)^\sigma$}
        \State{$check \leftarrow true$}
        \For{$(v,w)^\sigma \in \sim_a$ }
        \If{$w$ is not labelled with $\chi^\sigma$}
        \State{$check \leftarrow false$; \textbf{break}}
        \EndIf
        \EndFor

                \If{$check$}
        \State{label $v$ with $(K_a \chi)^\sigma$}
        \EndIf
        \EndCase

        \Case{$\psi^\sigma = ([\pi]^\bullet\chi)^\sigma$, where $\bullet \in \{+,-\}$}
        \If{$v$ is labelled with $\chi^{\sigma,[\pi]^\bullet}$}
        \State{label $v$ with $([\pi]^\bullet\chi)^\sigma$}
        \EndIf
        \EndCase

        \Case {$\psi^\sigma = ([\pi]^+)^\sigma $ with $[\pi]^+=  [\overline{\chi: \act \to A,B} ]^+_k $}
        %
        \For{$\ag \in \Ag$}
            \State{$X(a,v) = \{(v,w) \in \sim_a \mid (v,w) \text{ is labelled with } \sigma\}$}
        \EndFor
        \For{$i \in [1,...,k]$}
            \If{$\pos$ is labelled with $\chi_i^\sigma$}
                \For{$\ag \in A_i$}
                    \State{label $(a,v,\act_i)$ with $(\sigma, [\pi]^+; +)$}
                \EndFor
                
                \For{$a \in A_i \cup B_i$ and $(v,w)^\sigma \in \sim_a$}
                    
                    \If{not ($v$ is labelled $\chi_i^\sigma$ iff $w$ is labelled $\chi_i^\sigma$)}
                        \State{$X(a,v) \leftarrow X(a,v) \setminus (v,w)^\sigma$}
                    \EndIf
                \EndFor
            \EndIf
        \EndFor
        \For{$a \in \Ag$ and $(v,w)^\sigma \in X(a,v)$}
            \State{label $(v,w)$ with $\sigma, [\pi]^+$}
        \EndFor
        
        \For{$a \in \Ag$ and $\act \in \Act$}
        \If{$(a,v, \alpha)$ is not labelled with $(\sigma, [\pi]^+; +)$}
        \State{label $(a,v, \alpha)$ with $(\sigma, [\pi]^+; \pm(a,v, \alpha)^\sigma)$}
        \EndIf
        \EndFor
        \EndCase

        \Case {$\psi^\sigma = ([\pi]^-)^\sigma$ with $[\pi]^-= [\overline{\chi: \act \to A,B} ]^-_k$}
        \For{$a \in \Ag$}
            \State{$X(a,v) = \{(v,w) \in \sim_a \mid (v,w) \text{ is labelled with } \sigma\}$}
            \State{$Y(a,v) = \emptyset$}
        \EndFor
            \For{$i \in [1,...,k]$}
                \If{$\pos$ is labelled with $\chi_i^\sigma$}
                \For{$a \in A_i$}
                    \State{$Y(a,\pos) \gets Y(a,\pos) \cup \{\act_i\}$}
                \EndFor
                \For{$a \in A_i \cup B_i$ and $(v,w)^\sigma \in \sim_a$}
                        \If{not ($v$ is labelled $\chi_i^\sigma$ iff $w$ is labelled $\chi_i^\sigma$)}
                            \State{$X(a,v) \leftarrow X(a,v) \setminus (v,w)^\sigma$}
                        \EndIf
                \EndFor
                \EndIf
            \EndFor
        \For{$a \in \Ag$ and $(v,w)^\sigma \in X(a,v)$}
            \State{label $(v,w)$ with $\sigma, [\pi]^-$}
        \EndFor
        \State{$Z(a,v) = \{\act \mid (a,v,\act) \text{ is labelled } (\sigma; +)\}$}
        \If{$Z(\ag,\pos) \not \subseteq Y(a,\pos)$}
        \For{$(a,v,\act) \in Y(a,\pos)$}
        \State{label $(a,v,\act)$ with $(\sigma, [\pi]^-;-)$}
        \EndFor
        \EndIf
        \For{$a \in \Ag$ and $\act \in \Act$}
        \If{$(a,v, \alpha)$ is not labelled with $(\sigma, [\pi]^-; -)$}
        \State{label $(a,v, \alpha)$ with $(\sigma, [\pi]^-; \pm(a,v, \alpha)^\sigma)$}
        \EndIf
        \EndFor
        \EndCase

	\end{algorithmic}
\end{breakablealgorithm}

The algorithm follows the definition of the semantics. For the case $([\pi]^\bullet\chi)^\sigma$ (with $\bullet \in\{+,-\}$), we label the current state with $([\pi]^\bullet\chi)^\sigma$ if it is already labelled with $\chi^{\sigma,[\pi]^\bullet}$, i.e., if $\chi$ holds in the state after the sequence of updates $\sigma,[\pi]^\bullet$.

For the case $\psi^\sigma = ([\pi]^+)^\sigma$, we first create  a set $X(a,v)$ for each agent 
(line 13) to store the pairs $(v,w)\in\sim_a$ that are still present in the model after the sequence of updates $\sigma$. These sets keep track of how the observability for agent $a$ in state $v$ may change as a result of update  $[\pi]^+$. Then, for each element of $[\pi]^+$, if state $v$ satisfies precondition $\chi_i^\sigma$, we label triples $(a,v, \act_i)$ with $(\sigma, [\pi]^+; +)$ for each agent in $a \in A_i$, to denote that agent $a$ has action $\act_i$ available to her in state $v$ after the sequence of updates $\sigma, [\pi]^+$ 
(lines 14--17). Also, for a given element of $[\pi]^+$, for agents in $A_i \cup B_i$ we `cut' those pairs $(v,w)^\sigma$  for which $v$ and $w$ differ on the labelling of $\chi_i^\sigma$ 
(lines 18--20).
After we have went over elements of $[\pi]^+$, we label pairs $(v,w)^\sigma \in X(a,v)$ that `survived' the update with $\sigma, [\pi]^+$ 
(lines 21--22).
Actions of agents that are not affected by $[\pi]^+$ keep their $+$ or $-$ signs 
(lines 23--25).

The case $\psi^\sigma = ([\pi]^-)^\sigma$ is similar, with additional sets $Y(a,v)$ and $Z(a,v)$ to check that we do not remove all available actions from an agent. If we do not, then we label all the actions that are removed after the update with $(\sigma, [\pi]^-;-)$  
(lines 39--42),
and if we do, we, following the semantics, preserve all current actions of agent $a$.

Finally, with such a modification of the labelling algorithm, the rest follows the bottom-up approach from \cite{Schobbens04} described above with the only difference being that, given $(\coop{A}\varphi)^\sigma$, where $\varphi$ does not contain strategic modalities, we guess a uniform strategy w.r.t. to $\sim^\sigma$ such that for all agents $a\in A$, states $v$, and actions $\act$, the corresponding triple $(a,v,\act)$ is labelled with $(\sigma;+)$. In this way we ensure that the evaluation of $(\coop{A}\varphi)^\sigma$ takes into account the updates $\sigma$ that may have affected strategies of agents.

The preparation of $Sub(\psi)$ can be done in polynomial time w.r.t. $|\psi|$, and Algorithm $\ref{MC2}$ spends at most polynomial amount of time w.r.t. to $|\psi|$ and $|\iCGS|$. For the strategic cases, we call the \textsc{NP} oracle at most polynomial number of times.

\begin{theorem}
    Model checking \ATELD is $\Delta^{\mathrm{P}}_2$-complete. 
\end{theorem}

As a final remark of this section, we would like to mention that all the expressivity results for \ATLD hold, \textit{mutatis mutandis}, for \ATELD. This follows from the fact that in proofs of theorems in Section \ref{sec:expressivity}, every CGS can be considered as an iCGS with the identity observability relation. Thus, we have for \ATEL, \ATELD, \ATELD$^-$, and \ATELD$^+$ exactly the same expressivity results as in Corollary \ref{cor:exp}.  
 
\section{Related Work and Discussion}
\label{discussion}
\subsubsection{Normative Systems}
In Section \ref{sec:expressivity}, we showed how  $\ATLD^-$ captures some intuitions behind normative reasoning. Other related work on   normative updates includes \cite{ShohamT95,AgotnesHW10,AlechinaDL13,BullingD16,AlechinaG0P22,DitmarschKL23}. In comparison, our approach is more general and allows not only removing actions, but also adding them, as well as reasoning about the epistemic outcomes of such operations. 

\subsubsection{Dynamics of Ability}
Reasoning about updates of CGSs has been argued for in \cite{DynamicTurn}. 
The  updates studied in \cite{galimullin2025changing}, while being quite general, neither allow updating the set of agents' actions nor consider the imperfect information setting. Whereas in our logics the update of agents' abilities (in the sense of the power to force a certain outcome) may happen as a result of action updates, in \cite{GalimullinA21,GalimullinA22} the authors consider direct updates of coalitional powers. Thus, in \cite{GalimullinA21}, the authors consider the mechanism of granting \textit{new} dictatorial actions to a \emph{single} agent. In \cite{GalimullinA22}, the authors use more complex action models that allow redirecting transitions in a given CGS. However, action models can also create new states and hence the logic has a higher model checking complexity compared to \ATLD. 
In both works, strategic reasoning is restricted to the next-time fragment of \ATL, and thus lacks the tools for reasoning about extended consequences of updates.     

A series of works on \textit{obstruction logics} ($\mathsf{OL}$) (see, e.g., \cite{CattaLM23,CattaLMM24,DBLP:conf/ijcai/CattaLM025,AngelsDemons}) explores dynamic games between an agent that moves along a model and an agent or coalition that is able to modify the transitions of the model. Compared to various $\mathsf{OL}$'s, where model changes are implicit and in the scope of a quantifier, both \ATLD and \ATELD treat model updates as first-class citizens, which allows reasoning about different sequences of updates explicitly within a single formula. Finally, models of of $\mathsf{OL}$'s are weighted graphs and adding and removal of transitions are parametrised by their total cost. This is quite different from our setting, where we have no costs, and where adding and removal of transitions is parametrised by agents' actions.

\subsubsection{Capacity \ATL and \ATEL}
A recent line of work \cite{ballot2024strategic,ballot2025strategic,Ballot25} introduces Capacity \ATL and Capacity \ATEL in which models incorporate capacities as sets of actions for each agent. The logics can reason about  agent's knowledge of her capacity, but not about granting or revoking capacities.  
As the authors of \cite{ballot2025strategic} themselves suggest for future work, capacity mappings could be expressed by logical formulas, which is similar in spirit to our approach here.

Finally, compared to the related works discussed in this Section, $\ATELD$ also captures the epistemic consequences of granting and revoking actions, and also can express a variety of information-changing events from $\mathsf{DEL}$.

\subsubsection{A Note on Semantics} In our definition of the semantics of \ATELD, an agent is informed of a precondition $\varphi$ of an update $[\pi]^+$  even if she already has the action we are trying to grant. She is also informed about a precondition, even if an action removal is unsuccessful (see discussion after Definition \ref{def:ATLD-semantics}). 
We could have alternatively defined the semantics such that the agent is not explicitly informed of the precondition and instead \textit{infers} which state she is in as a result of receiving a new action. To capture this, we need to substitute $\sim_\varphi$ in Definition \ref{def:ATLED-semantics} with
$\sim_\varphi^{\act, a}:= (\llbracket \varphi \rrbracket_{\mathcal{G}_i} \times \llbracket \varphi \rrbracket_{\mathcal{G}_i}) \cup (\llbracket \lnot \varphi \rrbracket_{\mathcal{G}_i} \times \llbracket \lnot \varphi \rrbracket_{\mathcal{G}_i}) \cup (\llbracket has(a,\act) \rrbracket_{\mathcal{G}_i} \times  \llbracket has(a,\act) \rrbracket_{\mathcal{G}_i})$ for a given action $\act$. Then $\forall \ag \in \Ag, \sim_a' = \sim_a \cap \sim_{\varphi_i}^{\act_i,a}$ if  $a \in A_i \cup B_i$. With such a change, all of the technical results would also hold with minor modifications. An advantage of explicitly informing the agents of a precondition (as in Definition \ref{def:ATLED-semantics}) is that it allows us to express a greater range of epistemic updates, like public announcements that cannot be generally captured by the alternative `inferring' semantics.

\subsubsection{Discussion} 
We have introduced logics $\ATLD$ and $\ATELD$ to reason about the dynamics of granting and revoking actions to and from agents in perfect and imperfect information settings. We have studied the logics'  expressivity relative to their fragments and  $\ATL$ and $\ATEL$. We have also demonstrated that the model checking problems for $\ATLD$ and $\ATELD$ are \PTIME-complete and $\Delta_2^\textsc{P}$-complete respectively, i.e. no harder than for $\ATL$ and $\ATEL$. For $\ATLD$, we have also considered the bounded update existence problem and proved that it is $\textsc{NP}$-complete. 
%

There is a plethora of interesting directions for further research. The most immediate open technical questions are the satisfiabilty of the logics and the complexity of the update existence problem for \ATELD. We find it also tempting to study the validities of the logics and provide their axiomatisations. Already now we can point out that since the observability relation for $\ATELD$ is an equivalence, we get all the validities of the classic epistemic logic \cite{HalpernBook}. We also have, e.g., the validity $K_a has(\ag,\act) \leftrightarrow has(\ag,\act)$, which follows directly from the definition of an \iCGS. We can additionally consider group knowledge operators like distributed and common knowledge \cite{HalpernBook} to reason, e.g., in the case of distributed knowledge about the situations when agents can pool their knowledge together to infer which actions they collectively lack. It is also intriguing to relax the assumption that we model \textit{knowledge} of agents in a system and consider \textit{beliefs} instead. Then, we can explore the application of belief revision, including the classic AGM theory \cite{AlchourronGM85}, and, in particular, DEL-style belief revision \cite{FenrongBook}, to our setting.

On a conceptual level, based on our work, we can move towards \textit{proactively learning agents}~\cite{grosinger2025next}, where an agent could decide on her own which actions to learn (grant to herself) and which to unlearn (remove) when no longer needed. Another interesting direction is to explore action updates engendered by agents themselves, capturing the situations when agents with different abilities (possibly on different levels of a hierarchy) can grant each other actions or prohibit other agents' undesirable behaviours. Moreover, in our setting we assumed that all actions possible in a given system are specified in advance. An intriguing avenue of further research is to allow granting \emph{brand new} actions to agents in such a way that goes beyond the dictatorial setting of \cite{GalimullinA21}.
Finally, we also find it interesting to chart the full expressivity landscape of dynamic strategic logics from the literature by providing a comprehensive formal comparison of \ATLD and \ATELD with the logics mentioned above. 

\appendix

\bibliographystyle{kr}
\bibliography{refs.bib}

\clearpage
\appendix
\section*{Technical Appendix}

\thmtwo*

\begin{proof}
    Take formula $[\top:\beta \to \{a\}]^- \coop{\{a\}}\X \lnot p$ of $\ATLD^-$, and assume for a contradiction that there are equivalent $\psi$ of \ATL and $\chi$ of $\ATLD^+$. We take two CGS defined over the single agent $a$ and two actions $\alpha$ and $\beta$. The first CGS is $\CGS_3$ depicted in Figure~\ref{fig:proof1}, and the second CGS is $\CGS_4$, which is exactly like $\CGS_3$ with actions $\alpha$ and $\beta$ swapped.

    Now, we have that $\CGS_3, v \not \models [\top:\beta \to \{a\}]^- \coop{\{a\}}\X \lnot p$, as implementing the update $[\top:\beta \to \{a\}]^-$ results in $\CGS_2$, with $\beta$-transitions disabled, and from state $v$ the agent cannot reach the $\lnot p$-state $w$. 
    On the other hand, $\CGS_4, v \models [\top:\beta \to \{a\}]^- \coop{\{a\}}\X \lnot p$, because the update will remove $\beta$-transitions that are self-loops, and the agent will still be able to reach the $\lnot p$-state (recall that $\CGS_4$ is like $\CGS_3$ with $\alpha$ and $\beta$ swapped). 

    Since \ATL does not refer to action labels in its syntax, $\CGS_3$ and $\CGS_4$ are indistinguishable for \ATL formulas\footnote{CGSs $\CGS_3$ and $\CGS_4$ are, in fact, in the relation of alternating bisimulation \cite{agotnes07irr}.}. Thus we have that $\CGS_3, v \models \psi$ iff $\CGS_4, v \models \psi$, which implies $\ATLD^- \not \preccurlyeq \ATL$. Since \ATL is a fragment of $\ATLD^-$, we conclude that \ATL $\prec$ $\ATLD^-$.

    Finally, note that both $\CGS_3$ and $\CGS_4$ have all their transitions enabled, i.e. there are no more actions to grant to agent $a$ in both models. Therefore, $\CGS_3, v \models \chi$ iff $\CGS_4, v \models \chi$, and $\ATLD^- \not \preccurlyeq \ATLD^+$.
\end{proof}

\synthesis*

\begin{proof}

To demonstrate that the problem is \textsc{NP}-hard, we employ a reduction from the \textsc{NP}-complete 3SAT problem. Let $\Psi = \bigwedge_{1\leqslant i \leqslant m} X_i$ be an instance of 3SAT with $X_i = (x^i_1 \lor x^i_2 \lor x^i_3)$, where $x^i$ are literals. Also, let $Prop = \{p_1, ..., p_k\}$ be the set of propositional atoms occurring in $\Psi$. 

Given $\Psi$, we construct a CGS $\CGS_\Psi$ over a single agent $a$, the set of actions $\{\alpha, \beta, \beta_1, \beta_2, \beta_3, \gamma\}$, and the set of propositions $\{p, p_1, ...,p_k, p_{\top}, p_{\bot}\}$. The set of states $\setpos$ is $\{X_1, ..., X_m,$ $x^1_1, x^1_2, x^1_3, ..., x^m_1, x^m_2, x^m_3,$ $u_1, ..., u_k,$ $t_1, t_2\}$, where states $X_i$ correspond to SAT clauses, states $x^i_1, x^i_2, x^i_3$ correspond to literals for each clause, states $u_j$ correspond to propositional atoms, and states $t_1$ and $t_2$ correspond to the truth-values each atom can have. 
The labelling function $\val$ is defined as follows. For $X_i$, $\val(X_i) = \{p\}$; for $x^i_j$, $\val(x^i_j) = \{p_{\top}\}$ if $x^i_j$ is a positive literal, and $\val(x^i_j) = \{p_{\bot}\}$ otherwise; for $u_j$, $\val(u_j) = \{p_j\}$; and $\val(t_1) = \{p_{\top}\}$ and $\val(t_2) = \{p_{\bot}\}$. As for the availability function, $\availability(a, X_i) = \{\alpha,\beta_1, \beta_2, \beta_3\}$, $\availability(a, x^i_j) = \{\alpha\}$, $\availability(a, u_j) = \{\gamma\}$, and  $\availability(a, t_j) = \alpha$. 

Finally, we specify the transition function. For states $X_i$, $o(X_i, \alpha) = X_{i+1}$ for $1 \leqslant i < m$, and $o(X_m, \alpha) = X_m$. Intuitively, from $X_{i}$ the agent can access state $X_{i+1}$ via action $\alpha$, and for the last state $X_m$ we have a self-loop labelled with $\alpha$. The $\alpha$-transitions between $X_i$ states will model the choice of the clause to satisfy. For each $X_i$, we also have $o(X_i, \beta_j) = x^i_j$, which intuitively means that for a given clause $X_i$, the agent can choose one of the three literals by an appropriate $\beta_j$-action for $j \in \{1,2,3\}$. For both unavailable actions $\beta$ and $\gamma$, the transitions are self-loops, i.e. for all $X_i$, $o(X_i, \beta) = X_i$ and $o(X_i, \gamma) = X_i$.

For states $x^i_j$, $o(x^i_j,\alpha) = u_l$ such that $p_l$ is the atom appearing in literal $x^i_j$. For all other actions, which are unavailable, the transitions are self-loops. For states $u_l$, we have that $o(u_l, \gamma) = u_l$, i.e. self-loops. The same for the unavailable actions $\beta_1, \beta_2, \beta_3$. However, for the unavailable action $\alpha$ we have $o(u_l, \alpha) = t_1$, and for the unavailable action $\beta$ we have $o(u_l, \beta) = t_2$. We will use the $\alpha$- and $\beta$-transitions from $u_l$ to assign one of the truth-values to the atom $p_l$. Finally, all transitions from $t_1$ and $t_2$ are self-loops. 

As an example, take $\Psi = (p_1 \lor p_2 \lor p_3) \land (\lnot p_1 \lor \lnot p_3 \lor p_4)$. The corresponding CGS $\CGS_\Psi$ is depicted in Figure \ref{fig:proof2}.

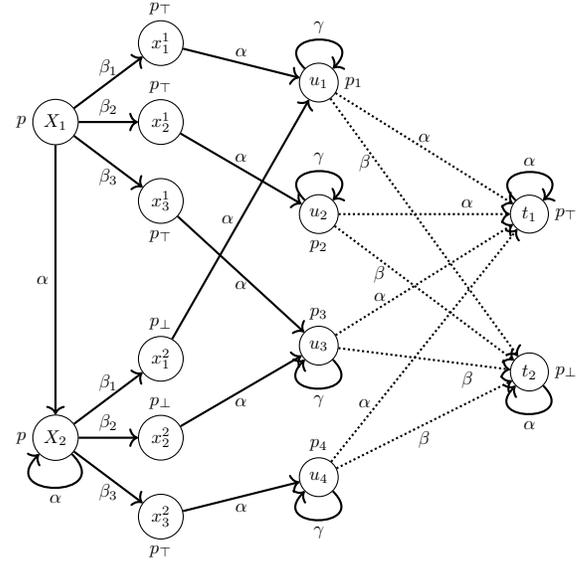
\begin{figure}[h!]
\centering
   \begin{tikzpicture}[scale=0.7, transform shape]
\node[circle,draw=black](v1) [label=left:$p$]at (0,0) {$X_1$};
\node[circle,draw=black](v2) [label=left:$p$] at (0,-6) {$X_2$};

\node[circle,draw=black](w11) [label=above:$p_\top$] at (2,1.5) {$x_1^1$};
\node[circle,draw=black](w12) [label=above:$p_\top$] at (2,0) {$x_2^1$};
\node[circle,draw=black](w13) [label=below:$p_\top$] at (2,-1.5) {$x_3^1$};

\node[circle,draw=black](w21) [label=above:$p_\bot$] at (2,-4.5) {$x_1^2$};
\node[circle,draw=black](w22) [label=above:$p_\bot$] at (2,-6) {$x_2^2$};
\node[circle,draw=black](w23) [label=below:$p_\top$] at (2,-7.5) {$x_3^2$};

\node[circle,draw=black](u1)[label=right:$p_1$]  at (5,0.75) {$u_1$};
\node[circle,draw=black](u2) [label=below:$p_2$] at (5,-1.75) {$u_2$};
\node[circle,draw=black](u3) [label=above:$p_3$] at (5,-4.25) {$u_3$};
\node[circle,draw=black](u4) [label=above:$p_4$] at (5,-6.75) {$u_4$};

\node[circle,draw=black](t1) [label=right:$p_\top$] at (9,-1.75) {$t_1$};
\node[circle,draw=black](t2) [label=right:$p_\bot$] at (9,-4.75) {$t_2$};

\draw [<-,thick](v2) to [loop below, out=225, in=315, looseness = 5] node[below] {$\alpha$} (v2); 
\draw[->,thick] (v1) to node[left] {$\alpha$}  (v2);

\draw[->,thick] (v1) to node[above] {$\beta_1$}  (w11);
\draw[->,thick] (v1) to node[above] {$\beta_2$}  (w12);
\draw[->,thick] (v1) to node[below] {$\beta_3$}  (w13);

\draw[->,thick] (v2) to node[above] {$\beta_1$}  (w21);
\draw[->,thick] (v2) to node[above] {$\beta_2$}  (w22);
\draw[->,thick] (v2) to node[below] {$\beta_3$}  (w23);

\draw[->,thick] (w11) to node[above] {$\alpha$}  (u1);
\draw[->,thick] (w12) to node[above] {$\alpha$}  (u2);
\draw[->,thick] (w13) to node[below] {$\alpha$}  (u3);

\draw[->,thick] (w21) to node[left] {$\alpha$}  (u1);
\draw[->,thick] (w22) to node[below] {$\alpha$}  (u3);
\draw[->,thick] (w23) to node[below] {$\alpha$}  (u4);

\draw[->,thick,densely dotted] (u1) to node[above] {$\alpha$}  (t1);
\draw[->,thick,densely dotted] (u1) to node[left, near start] {$\beta$}  (t2);
\draw[->,thick,densely dotted] (u2) to node[near end, above] {$\alpha$}  (t1);
\draw[->,thick,densely dotted] (u2) to node[near start, below] {$\beta$}  (t2);
\draw[->,thick,densely dotted] (u3) to node[near start, above] {$\alpha$}  (t1);
\draw[->,thick,densely dotted] (u3) to node[near end, below] {$\beta$}  (t2);
\draw[->,thick,densely dotted] (u4) to node[near start, left] {$\alpha$}  (t1);
\draw[->,thick,densely dotted] (u4) to node[below] {$\beta$}  (t2);

\draw [<-,thick](u1) to [loop above, out=45, in=135, looseness = 5] node[above] {$\gamma$} (u1);
\draw [<-,thick](u2) to [loop above, out=45, in=135, looseness = 5] node[above] {$\gamma$} (u2);
\draw [<-,thick](u3) to [loop below, out=225, in=315, looseness = 5] node[below] {$\gamma$} (u3);
\draw [<-,thick](u4) to [loop below, out=225, in=315, looseness = 5] node[below] {$\gamma$} (u4);

\draw [<-,thick](t1) to [loop above, out=45, in=135, looseness = 5] node[above] {$\alpha$} (t1);
\draw [<-,thick](t2) to [loop below, out=225, in=315, looseness = 5] node[below] {$\alpha$} (t2);

\end{tikzpicture}

\caption{CGS $\CGS_\Psi$ with the disabled transitions that are self-loops omitted for clarity. Disabled transitions that are not self-loops are depicted with dotted arrows.}
\label{fig:proof2}
\end{figure} 

Now, for a given $\Psi$, we construct the goal formula $\varphi_\Psi$ of \ATLD. Consider 
$$\varphi_\Psi = \coop{\{a\}}\G (p \land \coop{\{a\}}\X (\lnot p \land \varphi_{\true} \land \varphi_{\false}))$$
 with 
 $\varphi_{\true} = p_\top \to (\coop{\{a\}}\X (
 \lnot p_\top \land \coop{\{a\}}\X p_\top) \land \lnot \coop{\{a\}}\X \coop{\{a\}}\X p_\bot)$
  and $\varphi_{\false} = p_\bot \to (\coop{\{a\}}\X (\lnot p_\bot \land \coop{\{a\}}\X p_\bot) \land \lnot \coop{\{a\}}\X \coop{\{a\}}\X p_\top)$. 
  Intuitively, the goal formula $\varphi_\Psi$ means that for each clause (part $\coop{\{a\}}\G$) there is a literal such that we can assign unambiguously a truth-value to its atom that would satisfy the clause. `Unambiguously' here means that an atom cannot be assigned both \textit{true} and \textit{false} (parts $\varphi_{\true}$ and $\varphi_{\false}$). However, note that we do not have to assign a truth value to each atom to satisfy the goal formula, it is enough to have an unambiguous assignment to at least one atom to satisfy a clause.
  
A 3SAT instance $\Psi$ is satisfiable if and only if there is an update $[\pi]$ such that $\CGS_\Psi, X_1 \models [\pi] \varphi_\Psi$. Indeed, to satisfy $\Psi$ we need to satisfy at least one literal in each clause by setting the corresponding atoms to either \textit{true} or \textit{false}. Such an update would enable some of the unavailable actions from states $u_l$ such that if atom $p_j^i$ has to be set to \textit{true}, then action $\alpha$ will be enabled and action $\beta$  will remain unavailable. We proceed similarly for the case if $p_j^i$ has to be set to \textit{false}, but switching $\alpha$ and $\beta$. It is easy to verify that if $\Psi$ is satisfiable, then such $[\pi]$ exists, and if there is a $[\pi]$ that makes    $\CGS_\Psi, X_1 \models [\pi] \varphi_\Psi$ hold, then $\Psi$ is satisfiable.

For the example    $\Psi = (p_1 \lor p_2 \lor p_3) \land (\lnot p_1 \lor \lnot p_3 \lor p_4)$ depicted in Figure \ref{fig:proof2}, the corresponding update can be $[\pi]^+ = [p_2: \act \to \{a\}, p_4: \act \to \{a\}]^+$. Such an update will make the $\alpha$-transitions from states $u_2$ and $u_4$ available to the agent, and hence will allow the agent to reach $p_\top$ without allowing reaching $p_\bot$. This is equivalent to setting the truth-values of atoms $p_2$ and $p_4$ in $\Psi$ to \textit{true}. Hence, $\CGS_\Psi, X_1 \models [\pi]^+ \varphi_\Psi$. 

Finally, note that due to the construction of $\CGS_\Psi$, the size of such $[\pi]$ is linear in the size of $\Psi$. 
\end{proof}

\end{document}